\documentclass[11pt,english]{article}

\usepackage[utf8]{inputenc}
\usepackage{amsmath}
\usepackage{ltexpprt} 
\usepackage{makeidx} 
\usepackage{subfigure}
\usepackage{multirow}
\usepackage{graphicx}
\usepackage{epstopdf}
\usepackage{txfonts}
\usepackage{epsfig}
\usepackage{wrapfig}
\usepackage{url}
\usepackage{stfloats}
\usepackage{enumitem}
\usepackage{verbatim}
\setlist[itemize]{noitemsep}
\usepackage{clrscode3e}
\usepackage{listings}

\usepackage{color,soul}
\usepackage{colortbl}

\usepackage{palatino}

\def\astar{{A$^{*}$}}
\def\knn{{$k$-NN}}
\def\kNN{{$k$-NN}}

\date{\today}
\begin{document}
\title{SALT. A unified framework for all shortest-path query variants on road networks}


\author{
  Alexandros Efentakis\\  \small{Research Center ``Athena''} 
 \\ \small{efentakis@imis.athena-innovation.gr}
  \and 
 Dieter Pfoser\\ \small{George Mason University} 
 \\ \small{dpfoser@gmu.edu}
  \and 
 Yannis Vassiliou\\  \footnotesize{National Technical} \\ \small{University of Athens, Greece} 
 \\  \small{yv@cs.ntua.gr}
}

\maketitle

\begin{abstract}
Although recent scientific output focuses on multiple shortest-path problem definitions for road networks, none of the existing solutions does efficiently answer all different types of SP queries. This work proposes SALT, a novel framework that not only efficiently answers SP related queries but also $k$-nearest neighbor  queries not handled by previous approaches. Our solution offers all the benefits needed for practical use-cases, including excellent query performance and very short preprocessing times, thus making it also a viable option for dynamic road networks, i.e., edge weights changing frequently due to traffic updates. The proposed SALT framework is a deployable software solution capturing a range of network-related query problems under one ``algorithmic hood''.
\end{abstract}
%


\section{Introduction}
\label{sec:intro} 

During the last decades, recent scientific literature has focused on researching efficient methods for shortest-path (SP) related problems.
The related research 
has evolved so rapidly that even the recent overviews of \cite{delling09,bauer2010} had to be updated in subsequent publications \cite{bast2014}. 
Unfortunately, despite this plethora of efficient algorithms 
only \emph{few of them may actually be used in a practical application context}. 
The requirements to such a potent approach should be (i) preprocessing time of few \emph{seconds} for continental road networks and (ii) SP query times of a few \emph{ms}. 
Currently, only two candidates fit these strict requirements. The \emph{graph-separator} approach of Customizable Route Planning (CRP) \cite{delling2011,delling2013e} and the recent adaptation \cite{efentakis2013d} of the ALT~\cite{goldberg2005} algorithm. 
Due to these specific properties, said algorithms are used in commercial solutions, such as Bing Maps and SimpleFleet~\cite{efentakis2013c}, 
for their live traffic-based routing.

Unfortunately, most of the developed algorithms are tuned to solving a specific problem efficiently, but are rather inefficient when used in a different context. Contrarily,  engineering a framework that efficiently solves multiple shortest-path problems, would not only increase the commercial potential of such a solution but would also be the first step towards the direction of a \emph{grand unified SP toolkit}. 
To this purpose, Efentakis et al.~\cite{efentakis2014} extended graph-separators and proposed the novel set of GRASP (Graph separators, RAnge, Shortest Path) algorithms that solve most variants of the single-source shortest-path problems on road networks, including \emph{one-to-all} (finding SP distances from a source vertex $s$ to all other graph vertices), \emph{one-to-many} (computing the SP distances between the source vertex $s$ and all vertices of a set of targets $T$) and \emph{range queries} (find all nodes reachable from $s$ within a given timespan / distance). GRASP requires minimal preprocessing time and provides excellent parallel query performance needed in the context of practical applications and respective commercial solutions.

Another fundamental problem frequently encountered in location-based services is the \knn\ query, i.e., given a query location and a set of objects
on the road network, the \knn\ search finds the $k$-nearest objects to the query location. 
Unfortunately, 
even recent attempts, such as G-tree \cite{zhong2013}, are not scalable with respect to the network size, since they require preprocessing of \emph{several hours} for continental road networks.  
In addition, for a large number of randomly distributed objects, an efficient Dijkstra implementation could answer \knn\/ queries (for small values of $k$) by settling a few hundreds nodes and requiring $< 1ms$. 
Moreover, most existing approaches (contrarily to Dijkstra) require a \emph{target-selection phase}, i.e., they need to mark the objects location within the underlying index. This phase takes a few seconds, hence having limited appeal for applications involving moving objects (e.g., vehicles). 
Therefore, it only makes sense to use a complex (as in non-Dijkstra) \knn\ processing framework in cases of either rather ``small'' numbers of objects or objects following skewed distributions (e.g., POIs located near the city center), i.e.,~for cases in which Dijkstra does not perform well.

Putting everything together, the ambition of this work is to provide a unified algorithmic solution that may be used in a \emph{dynamic road network} context by having very short preprocessing times and competitive query times, while covering a \emph{wide range of shortest-path and network search problems}, such as (i) single-pair, (ii) one-to-all, (iii) one-to-many, (iv) range and (v) \knn\ queries.
Specifically, we aim at combining the fragmented approaches related to the various shortest-path problem definitions and instead propose a unified framework that tackles all of them. Our proposed \textbf{SALT} (graph Separators + ALT) framework requires only seconds for preprocessing continental road networks and provides excellent query performance for a wide range of problems. We will show that SALT is (i) $3-4 \times$ faster for point-to-point queries when compared to existing methods of similar preprocessing times, (ii) it answers one-to-all, one-to-many and range queries with comparable performance to state-of-the-art approaches, and most importantly, (iii) 
it may also answer \knn\ queries in $<1ms$, for both, static or moving objects. 
As such, our SALT framework could be a \emph{swiss-army-knife for tackling all shortest-path problem variants}, making it a serious contender for use in commercial applications.    

The outline of this work is as follows. Section~\ref{sec:related_work} describes relevant previous related work. Section \ref{sec:contribution} describes our novel SALT framework and algorithms. Experiments establishing the benefits of our approach are provided in Section \ref{sec:experiments}. Finally, Section~\ref{sec:conclusions} gives conclusions and directions for future work.    

\section{Related work}
\label{sec:related_work}

Throughout this work, we are dealing with directed weighted graphs $G(V,E,w)$, where $V$ is the set of vertices, $E \subseteq  V\text{x}V$ are the arcs of the graph and $w$ is a positive weight function $E \rightarrow R^{+}$. 
The reverse graph $\overline{G}= (V,E)$ is the graph obtained from $G$ by substituting each arc $(u,v) \in E$ by $(v,u)$. 

A partition of $V$ is a family of sets $C=\{c_{0},c_{1},\dots c_{M}\}$, such that each node $u \in V$ is contained in exactly one set $c_{i}$. An element of a partition is called a $cell$. A multilevel partition of $V$ is a family of partitions $\{C^{0},C^{1},\dots C^{L}\}$ where $\ell$ denotes the level of a partition $C^\ell$.  Similar to \cite{delling2011}, level 0 refers to the original graph, $L$ is the highest partition level and in this work we use \emph{nested multilevel partitions}, i.e.,~for each $\ell < L$ and each cell $c^\ell_i$ there exists a unique cell $c^{\ell+1}_j$ (called the supercell of $c^\ell_i$) with $c^\ell_i  \subseteq c^{\ell+1}_j$. Accordingly, $c^\ell_i$ is a subcell of $c^{\ell+1}_j$. In this notation, 
$c^\ell(v)$ is the cell containing the vertex $v$ on level $\ell$. Accordingly, the number of cells of the partition $C^\ell$ is denoted as $|C^\ell|$.  For a boundary arc on level $\ell$, the tail and head vertices are located in different level-$\ell$ cells; a boundary vertex on level $\ell$ is connected with at least one vertex in another level-$\ell$ cell. Note that for nested multilevel partitions, a boundary vertex/arc at level $\ell$ is also a boundary vertex/arc for all levels below.

In $k$-NN queries, given a query location $s$ and a set of objects $O$,
the $k$-NN search problem finds k-nearest objects to the query location. 
Similar to \cite{zhong2013}, we assume that both the query location and the objects are located at vertices, which is a logical compromise, since (i)~for static objects we can add new vertices at object locations or (ii)~in the case of moving objects we can assign objects to vertices heuristically, e.g., placing them on their nearest vertex. 

\subsection{Landmarks and the ALT algorithm.}
\label{subsec:alt_algor}

In the ALT algorithm \cite{goldberg2005}, a small set of vertices called landmarks is chosen. Then, during preprocessing, we precompute distances to and from every landmark for each graph vertex.  Given a set $S{\subseteq} V$ of landmarks and distances $d(L_{i},v)$, $d(v,L_{i})$ for all vertices $v{\in} V$ and landmarks $L_{i}{\in} S$, the following triangle inequalities hold: $d(u, v) {+} d(v,L_{i}) \geq d(u, L_{i})$ and $d(L_{i}, u) {+} d(u, v) \geq d(L_{i}, v)$ . Hence, the function $\pi_{f}=max_{L_{i}}max\{ d(u,L_{i})-d(v, L_{i}),d(L_{i},v)-d(L_{i},u)\}$  
provides a lower-bound for the graph distance $d(u,v)$.  
Later works \cite{Potamias2009} 
showed that landmarks may also be used for providing upper-bounds on the graph distance between any two vertices. 
Thus, 
landmarks provide a fast and efficient way to approximate graph distances, according to Eq.~\ref{eq:lower_bound} and \ref{eq:upper_bound}.

\begin{equation}
d(u,v) \geq max_{L_{i}}max\{ d(u{,}L_{i}) - d(v{,} L_{i}),d(L_{i}{,}v) - d(L_{i}{,}u)\}
\label{eq:lower_bound}
\end{equation}
\vspace{-2em}
\begin{equation}
d(u,v) \leq min_{L_{i}}(d(u,L_{i})+d(L_{i},v))
\label{eq:upper_bound}
\end{equation}

ALT combines the classic \astar\ algorithm \cite{A-star68} with the aforementioned lower-bounds. 
For  bidirectional search, ALT uses the average potential function \cite{ikeda1994} defined as $p_{f}(v)=(\pi_{f}(v) - \pi_{r}(v))/2$ for the forward and $p_{r}(v)=(\pi_{r}(v) - \pi_{f}(v))/2 = -p_{f}(v)$ for the backward search.

\vspace{-2pt}
\subsection{Graph separators}
\label{subsec:crp_algo}

In Graph Separator (GS) methods, such as CRP \cite{delling2011,delling2013e}, a partition $C$ of the graph is computed. Then, during preprocessing, we build an \emph{overlay graph} $H$ containing all boundary vertices and arcs of \emph{G}. It also contains a clique for each cell~\emph{c}: for every pair $(u,v)$ of boundary vertices in $c$, a clique arc $(u,v)$ is created whose cost is the same as the shortest path (restricted to the inner edges of~\emph{c}) between $u$ and $v$.  
For a SP query between $s$ and $t$, the Dijkstra algorithm must be run on the graph consisting of the union of $H$, $c^{0}(s)$ and $c^{0}(t)$. To further accelerate queries, we may use multiple levels of overlay graphs. 
Since each clique is calculated by using only the inner edges of \emph{c}, GS preprocessing may be easily parallelized. Moreover, overlay graphs of higher level partitions may be computed by using the overlay graphs of lower levels. 
By using those two optimizations, CRP is the most efficient SPSP algorithm in terms of preprocessing time (requiring few seconds for continental road networks) and is thus suitable for dynamic road networks. 

\vspace{-2pt}
\subsection{Single-source shortest-path queries}
\label{subsec:sssp-queries}

Recently, Efentakis et al.~\cite{efentakis2014} expanded graph separators and proposed GRASP, a novel set of algorithms for handling all variants of single-source shortest-path queries, including one-to-all, one-to-many and range queries. All three algorithms, namely GRASP (one-to-all), isoGRASP (range) and reGRASP (one-to-many) use the exact same data structures and share all the advantages of graph-separator methods, such as very short preprocessing times and excellent parallel query performance. Unfortunately, parallel reGRASP requires a few ms for one-to-many queries on continental road networks and hence is not fast enough for handling $k$-NN queries.
  
\vspace{-2pt}
\subsection{$k$-NN queries}
\label{subsec:$k$-NN-queries}

There are many works on $k$-NN queries for static objects on road networks. Unfortunately, 
even latest attempts, such as G-tree \cite{zhong2013}  cannot \emph{scale for continental road networks},
\emph{requiring 16.8 hours of preprocessing for the full USA network}. 
Moreover, previous index-based approaches require a \emph{target selection phase} to index which tree-nodes contain objects (a process requiring few seconds) and hence, they cannot be applied in case of moving objects. Thus, all previous solutions are quite unsuitable for practical applications and dynamic road networks. 

There is also significant work around $k$-NN queries for moving objects on road networks.  
Still, most of those approaches are either disk-based \cite{wang2011}, have not been tested on continental road networks \cite{jensen2003,mouratidis2006b,wang2011} and cannot address dynamic live-traffic road networks. 
Recently, CRP was also expanded~\cite{delling2013h} to handle $k$-NN queries. Unfortunately, it shares the limitations of previous methods, since (i) it also requires a target selection phase and therefore cannot be applied to moving objects and (ii) it may only perform well for objects near the query location (otherwise the whole upper level of the overlay graph must be traversed). Hence, this solution is also far from optimal.

\section{The SALT framework}
\label{sec:contribution}

The main contribution of this work is to propose SALT (graph Separators + ALT algorithm), a unified framework for answering point-to-point, single-source (one-to-all, one-to-many and range) and especially $k$-NN queries which are not handled efficiently by existing approaches. The main advantage of SALT is, that the exact same data structures may service all the different type of SP queries and hence, SALT may be easily integrated into commercial, real-world applications. What follows is a detailed discussion of the SALT framework.

\begin{figure*}[!tb]
\centering
 \subfigure[A sample graph $G$. $|V|{=15}$, $L{=}2$, $|C^L|{=}2$]
{\includegraphics[width=0.158\textwidth]{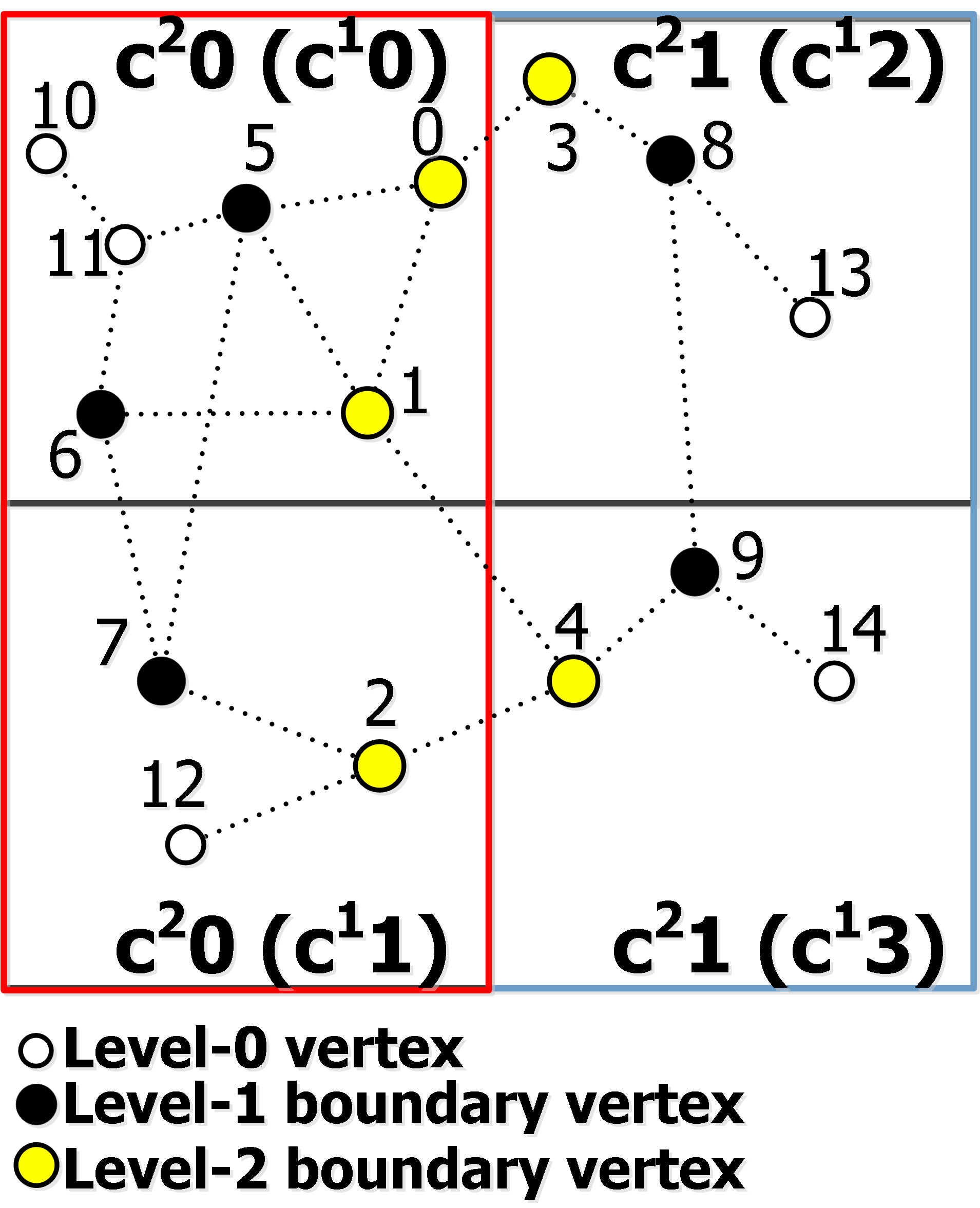}\label{fig:orig}}\quad
 \subfigure[Building the level-1 overlay graph]
{\includegraphics[width=0.158\textwidth]{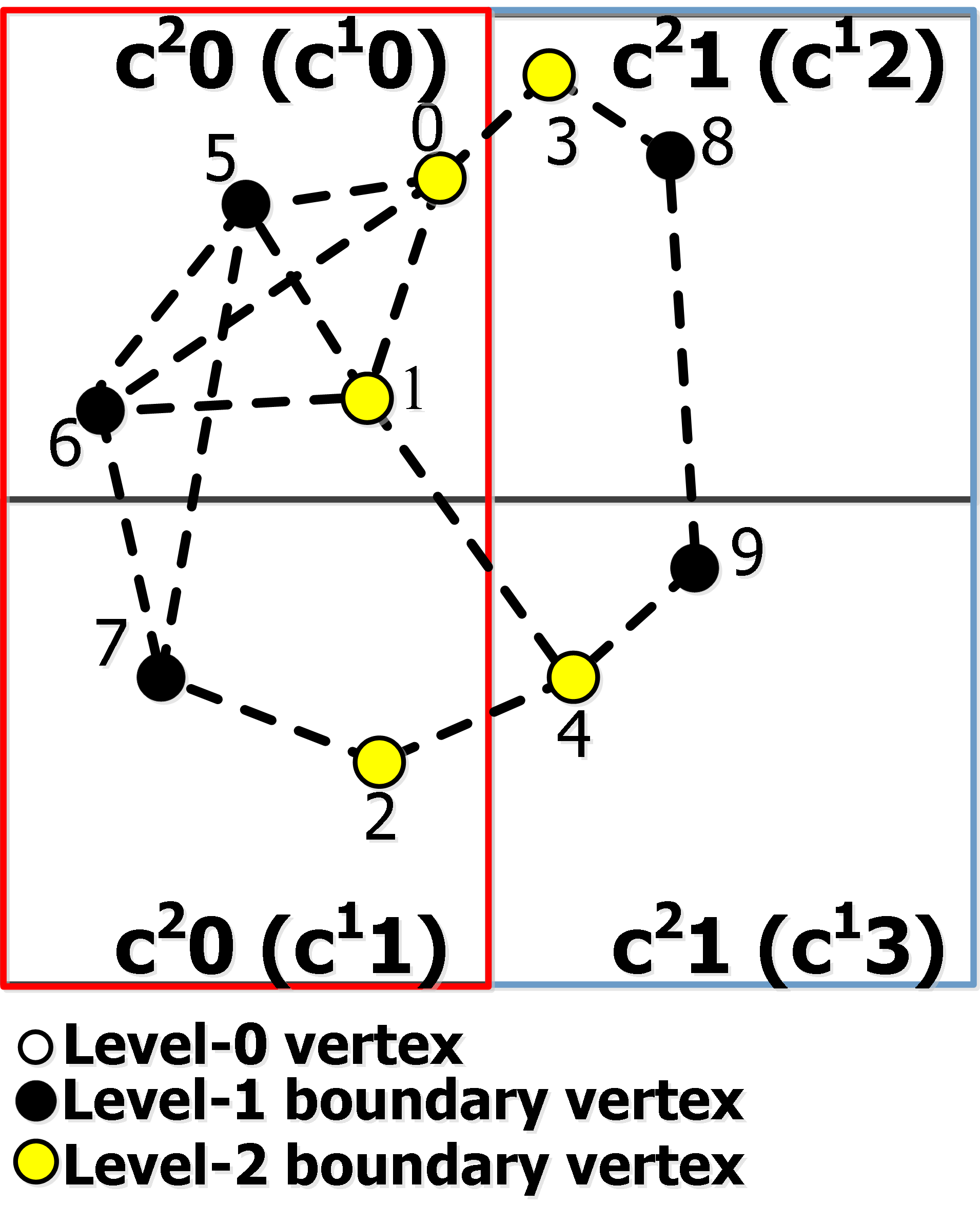}\label{fig:overlay1}}\quad
 \subfigure[Building the level-2 overlay graph]
{\includegraphics[width=0.158\textwidth]{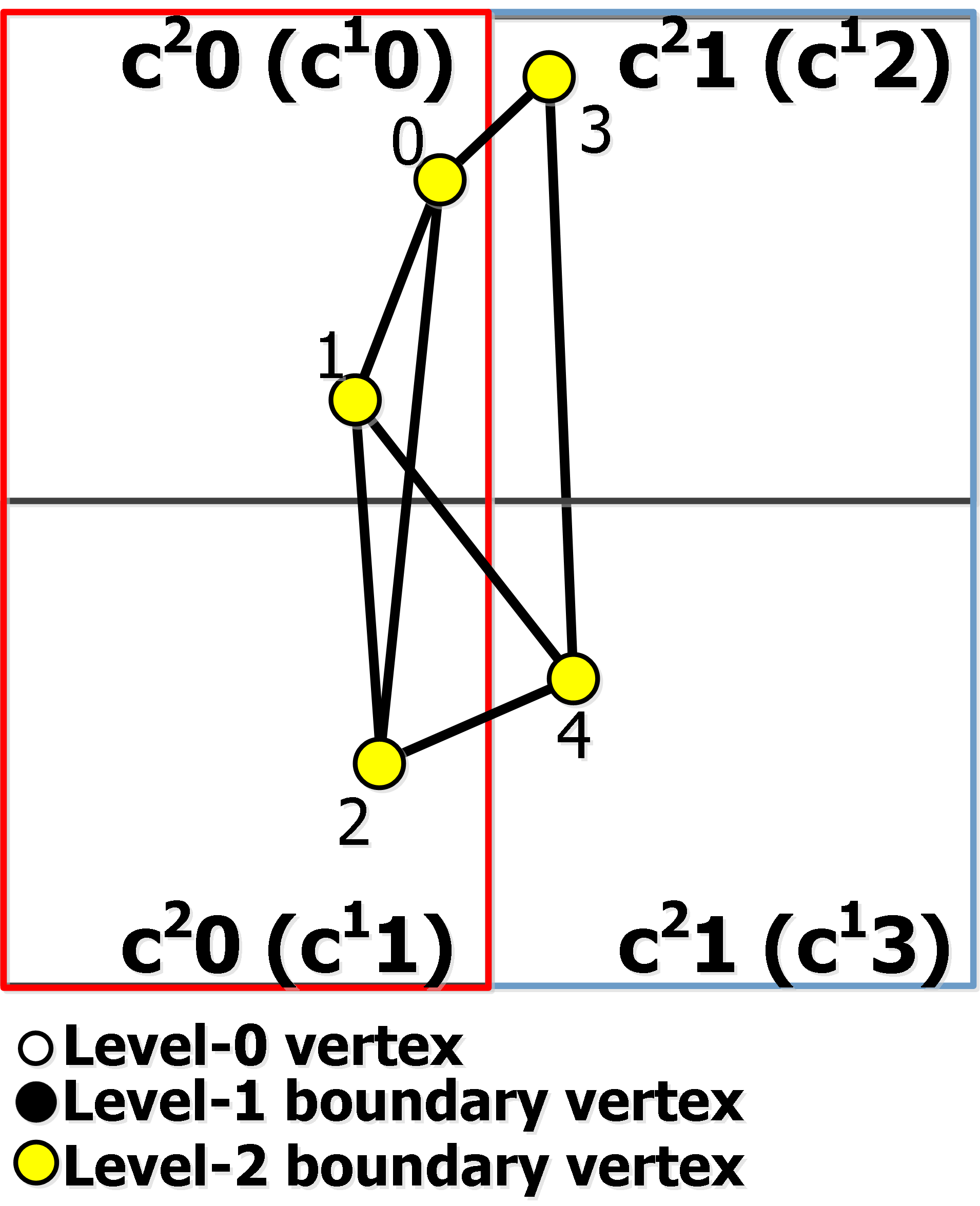}\label{fig:overlay2}}\quad
 \subfigure[Downward arcs and the $G_{GS}{\downarrow}$ graph]
{\includegraphics[width=0.158\textwidth]{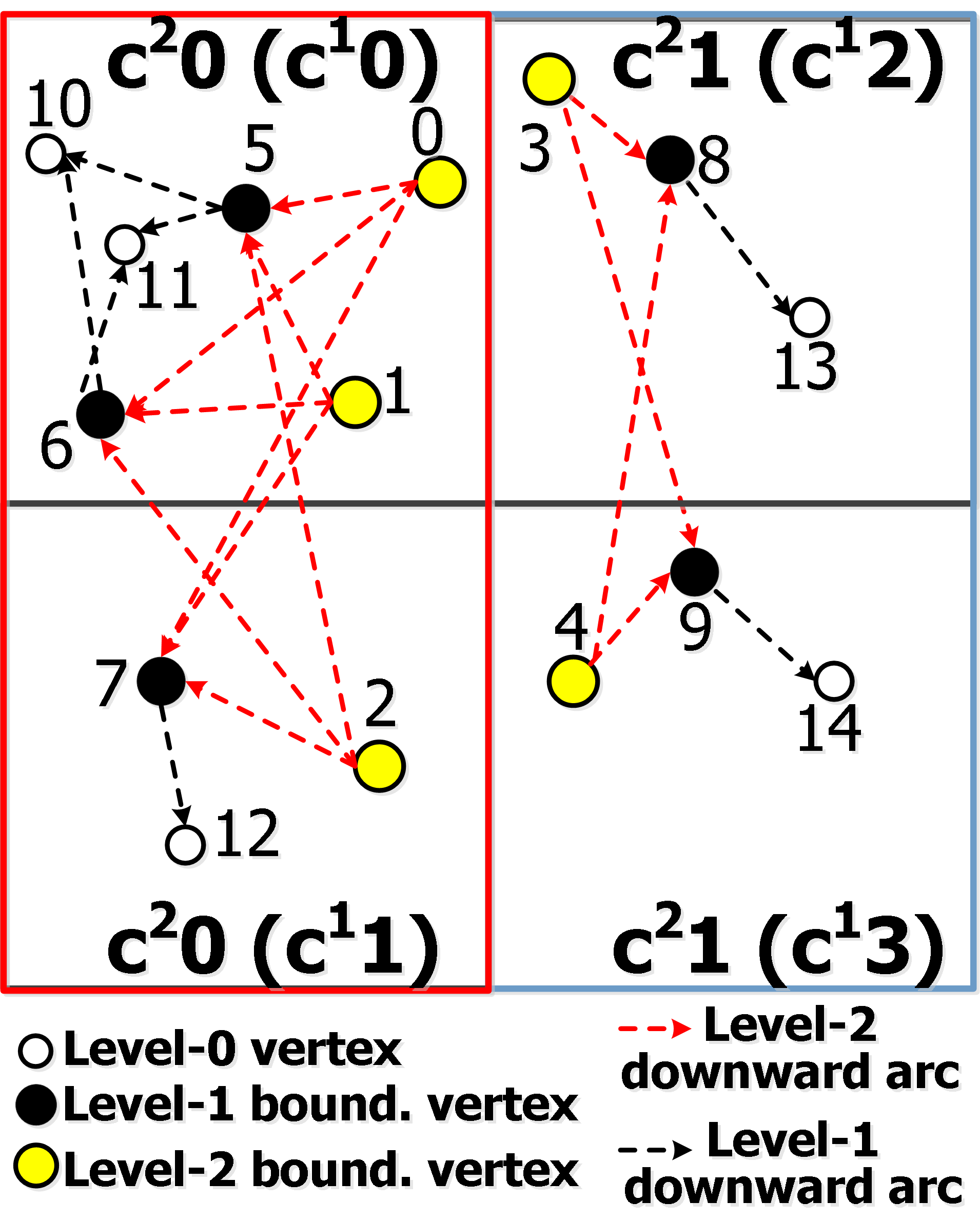}\label{fig:downward_graph}}
 \caption{SALT's GS customization phase. Building the overlay graph $H$ and the $G_{GS}{\downarrow}$ graph}
  \label{fig:GRASP_graphs}
\end{figure*}

\subsection{Preprocessing}
\label{subsec:preprocessing}

SALT's preprocessing consists of two distinct phases, (i) the \emph{graph-separator (GS) phase} and (ii) the \emph{landmarks preprocessing phase}.

The \emph{GS phase} of SALT mimics the preprocessing of GRASP \cite{efentakis2014} (see Fig.~\ref{fig:GRASP_graphs}). During this phase, we use the Kafpaa/Buffoon \cite{sanders2011e} partitioning tool to create nested multilevel partitions of the road network graph in a top-down fashion. This initial partitioning phase is \emph{metric independent} and needs to be executed only once, i.e., even in the case of arc-weights changes or for different metrics.
Following partitioning, the \emph{customization stage} builds the overlay graph $H$ containing all boundary vertices and arcs of $G$. The graph $H$ also contains a clique for each cell $c$: for every pair $(u,v)$ of boundary vertices in $c$, we create a shortcut arc $(u,v)$ whose cost is the same as the shortest-path (restricted to inner edges of $c$) between $u$ and $v$ (see Fig.~\ref{fig:overlay1},~\ref{fig:overlay2}).  
Similar to \cite{efentakis2014}, we also calculate the SP distances between all border vertices of level $\ell$ and all vertices of level $\ell{-}1$ within each cell $c^\ell$ (see Fig.~\ref{fig:downward_graph}). To differentiate between the two kinds of arcs computed, we will denote as (i) \emph{clique arcs} the added overlay arcs that connect border vertices of the same level $\ell$ and (ii) \emph{downward arcs} of level $\ell$ the vertices connecting different levels, i.e., $\ell$ and $\ell{-}1$.  
For added efficiency, \emph{downward arcs} are stored as a separate graph, 
referred to, as $G_{GS}{\downarrow}$. Both types of arcs are computed bottom-up and starting at level one. To process a cell, the GS customization stage for SALT executes a Dijkstra algorithm from each boundary vertex of the cell. 
We also apply the arc-reduction optimization of \cite{efentakis2012}, which reports only distances of boundary vertices that are direct descendants of the root of each executed Dijkstra algorithm. 

Although SALT's GS preprocessing phase is conceptually similar to GRASP, there are two major differences. 
(i)~For accelerating SALT's preprocessing, $H$ and $G_{GS}{\downarrow}$ have the same number of levels ($L=6$ in our experiments) with $|C^L| =16$ (cf. the original GRASP paper with $|C^L| =128$ and $L=16$). 
Using a smaller number of cells at the upper level of the cell hierarchy slightly lowers one-to-all query parallel performance, but \emph{accelerates point-to-point queries and reduces preprocessing time}. Hence, it is a very logical compromise, since our focus is on increased versatility. 
(ii)~Moreover, we have to repeat SALT's GS customization stage twice, one for the forward and one for the reverse graph. This is necessary for the landmarks phase of SALT, but it also allows to answer, both, forward and reverse single-source queries. Thus, at the end of SALT's GS preprocessing we have built two versions of the overlay graphs, $H$ and $G_{GS}{\downarrow}$, one for forward and one for reverse graph queries, respectively.  

The \emph{landmarks preprocessing phase} for SALT extends the preprocessing proposed by~\cite{efentakis2013d}, which  optimized and tailored the ALT algorithm for use with dynamic road networks. Landmarks are selected by the \emph{partition - corners} landmarks selection strategy, in which the we use the cells created by Kafpaa and from each cell we select the four corner-most vertices as landmarks. 
For SALT, we accelerate the computation of distances of all graph vertices from and to landmarks by executing two sequential GRASP algorithms (forward and reverse) instead of using plain Dijkstra (as in all previous approaches). 
Moreover, we may perform those $2 {\times} |S|$ GRASP algorithms in parallel. By using these optimizations, the landmarks preprocessing phase of SALT never takes more than \emph{$4s$ for 24 landmarks} and is therefore \emph{at least $6\times$ faster than any existing work}.

Conclusively, at the end of the preprocessing stage of SALT, we have built the overlay graphs $H$ and $G_{GS}{\downarrow}$ for both forward and reverse searches and calculated distances for all vertices from and to the selected landmarks. 
This is all the information required for answering point-to-point, single-source and $k$-NN queries. For \emph{dynamic road networks},
we only need to repeat the GS customization stage and the computation of distances of all vertices from and to the landmarks. Both these phases \emph{require less than $19s$} for the benchmark road networks we used. This makes SALT suitable for dynamic scenarios. 

\subsection{Single-pair shortest-path queries}
\label{subsec:p2p_sp_queries_contrib}

Using the SALT preprocessing data, we can accelerate point-to-point shortest-path queries by combining our custom CRP (with arc-reduction) with the goal-direction technique of the ALT algorithm.  In CRP, to perform a SP query between $s$ and $t$, Dijkstra's algorithm must be run on the graph consisting of the union of $H$, $c^{0}(s)$ and $c^{0}(t)$. The difference in SALT is that, instead of Dijkstra, we use the ALT algorithm on the graph consisting of the union of  $H$, $c^{0}(s)$ and $c^{0}(t)$. Note that both ALT and CRP may also be used in, either, a unidirectional or a bidirectional setting. A similar combination of CALT \cite{bauer2010} and CRP was unofficially introduced in \cite{delling2011}, which uses the landmark derived lower-bounds only on the upper-level of the graph-separator overlay graph. Therefore, \emph{local searches} could not be accelerated. 
Local search is crucial for \knn\ queries, since the \knn\ results for small values of $k$ are usually located close to the query location.  In contrast, our \emph{SALT-p2p algorithm}, combining pure ALT and SIMD instructions for lower bound calculations
and our custom CRP, is going to be significantly more efficient than stand-alone ALT or CRP. 
In addition, since both methods are extremely robust to the metric used \cite{bauer2010,delling2011}, their combination will provide excellent performance for both travel times and travel distances.

\begin{theorem}
The SALT-p2p algorithm for SPSP queries is correct.
\end{theorem} 

\begin{proof}
Since the SALT-p2p algorithm for SPSP queries combines CRP and ALT, which both provide optimal results then the SALT-p2p algorithm also provides correct results.
\end{proof} 

\subsection{$k$-NN queries}
\label{subsec:$k$-NN_queries_contrib}
 
SALT's preprocessing data can be used to answer \knn\ queries. 
Instead of initiating a \knn\ search from a query location $s$ to objects $O$, we start a search from all the objects at the same time \emph{to} the query location in the reverse graph. Hence, we take advantage of, both, GS and ALT acceleration to guide the search towards the query location. 
The SALT-kNN algorithm's query phase may be divided in two independent stages. 
The \emph{Pruning phase} excludes objects that cannot possibly belong to the \knn\ set by using the upper and lower-bounds provided by the landmarks preprocessing data.
The \emph{Main phase} executes a unidirectional SALT-p2p algorithm in the reverse graph from all remaining objects at the same time to the query location until the query location is settled. 
Now we have found the first nearest neighbor. This process has to be repeated another $k-1$ times until all $k$-NN are discovered.
The algorithm is detailed in the following.

\vspace{5 pt}
\textbf{Pruning phase}. To prune objects that are too far away from the query location and thus cannot belong to the k-nearest neighbors set, we 
(i)~calculate the $k$-th lowest upper-bound of graph distances between the query location and the objects (cf. Equation~\ref{eq:upper_bound}) and 
(ii)~pruning/exclude objects whose distance lower-bounds between them and the query location  (cf. Equation~\ref{eq:lower_bound}) exceed the $k$-th lowest upper-bound. To the best of our knowledge, this is the \emph{first work to utilize upper and lower landmark bounds in the context of \knn\ queries}. 

\begin{theorem}
The pruning phase of the SALT-kNN algorithm is correct.
\end{theorem} 

\begin{proof}
When we calculate the $k$-th lowest upper-bound of distances between the query location and the objects, we can guarantee that there are at least $k$-neighbours within this distance from the query location. So, any object located farther than that (as provided by the landmarks provided lower-bounds) may be safely pruned.  
\end{proof}

To accelerate the process of computing the $k$-th lowest upper-bound between the query location and the objects we can use a bounded max-heap $Q$ that only stores $k$-upper-bounds and procedure \proc{getKthLowestUpperBound}: 

\noindent{
\begin{minipage}{1\textwidth}

\begin{minipage}[t]{.5\textwidth}
{\small
\begin{codebox}
\Procname{$\proc{getKthLowestUpperBound}(s, O)$}
\li $Q \gets emptyMaxHeap$
\li $m \gets 0$
\li \For each $obj$ in $O$ 
\li \Do
\If $m < k$
\li \Then
$Q.push(upperBoundDist(s,obj))$ 
\li $m \gets m+1$
\li \ElseIf ($upperBoundDist(s,obj) < Q.top()$)
\li \Then
$Extract-max(Q)$
\li $Q.push(upperBoundDist(s,obj))$
\End
\End
\li \Return $Extract-max(Q)$
\End
\end{codebox}
}
\end{minipage}
\begin{minipage}[t]{.5\textwidth}
{\small
\begin{codebox}
\Procname{$\proc{PruningPhase}(s, O)$}
\li $\id{O_{small}} =\{\}$
\li $\id{kthUpperBound} \gets getKthLowestUpperBound(s, O) $
\li \For each $obj$ in $O$ 
\li \Do
\If $\id{lowerBoundDist(s,obj)}  \leq \id{kthUpperBound} $
\li \Then
$\id{O_{small}}.add(obj)$
\End
\End
\li \Return \id{O_{small}}
\End
\end{codebox}
}
\end{minipage}
\end{minipage}

\vspace{5pt}
Since the  bounded max-heap $Q$ only stores $k$-upper-bound distances, we only need to compare the next objects's upper-bound with the top of the heap. If we have found a lower upper-bound, we remove the top of the heap and add the new upper-bound to $Q$. 
At the end of the procedure, the top of the max-heap is the $k$-th lowest upper bound of distances between the query location and the objects. 
The pruning phase is now implemented by procedure \proc{PruningPhase}.

At the end of the pruning phase, instead of using the objects in $O$, we only need to check for the $k$-nearest neighbors within the objects in $O_{small}$. Our experimentation has shown that the pruning phase is very effective, since it efficiently prunes more than $60\%$ of the total number of objects in $O$.   

\begin{wrapfigure}{l}{0.5\textwidth}
\vspace{-10pt}
{\noindent
\begin{minipage}[t]{.5\textwidth}
{\small
\begin{codebox}
\Procname{$\proc{MainPhase}(s, \id{O_{small}}, k, \overline{G})$}
\li \For $i \gets 0$ \To $k{-}1$
\li \Do
$T'=$new vertex
\li \For each $obj \in \id{O_{small}}$ 
\li \Do
Connect $T'$ to $obj$ with zero-weight edges 
\End
\li $(\id{iNN}, \id{iNNdist})  \gets SALT{-}p2p(T'{,}s,\overline{G})$
\li $\id{O_{small}} \gets \id{O_{small}} - \id{iNN}$
\End
\End
\end{codebox}
}
\end{minipage}
}
\end{wrapfigure}

\vspace{5 pt}
\textbf{Main phase}. Following the pruning phase, to find the first nearest neighbor we start by performing a search simultaneously from all objects in $O_{small}$ to the query location in the reverse graph. 
To do so, we use the idea of \cite{maue2006}. We add a new vertex $T'$ connected to all objects  in $O_{small}$ using zero-weight edges and then perform a unidirectional SALT-p2p algorithm from $T'$ to the query location $s$ in the reverse graph (see Figure~\ref{fig:salt-knn}). 
At the end of this process, we have found the first nearest neighbor of query location $s$. Then we eliminate this vertex from $O_{small}$ and repeat the process for another $k{-}1$ iterations to retrieve the full \knn\ set (see procedure \proc{MainPhase}).

\begin{theorem}
The main phase of the SALT-kNN algorithm is correct.
\end{theorem}

\begin{proof}
As \cite{maue2006} has shown, by adding a new vertex $T'$ connected to a set of vertices and then running any correct shortest-path algorithm from  $T'$ to a destination vertex, we can find the minimum shortest-path distance between  the vertices set and the destination. As a result, since the SALT-p2p algorithm is correct and this algorithm is run in the reverse graph, at the end of the first iteration we have found the shortest-path distance between the query location~$s$ and one of the objects in $O_{small}$, which is the first nearest-neighbour. Since, we eliminate this vertex from $O_{small}$ and repeat this correct process for another $k{-}1$ times, the main phase of the SALT-kNN algorithm is correct.        
\end{proof}

To retrieve not only the shortest-path distance between the query location $s$ and the objects in $O_{small}$, but also the actual \knn\ vertices, 
we need to maintain for each labeled vertex a reference that points to the originating vertex in the objects' set $O_{small}$. Thus, when we extract the query location~$s$ from the priority queue and terminate the SALT-p2p algorithm at the $i$-th iteration, we know not only the $i$-th shortest-path distance but the $i$-th nearest neighbor as well.
Moreover for each object $o$ in $O_{small}$, we need to store the cell ID $c^{1}(o)$ of the cell this object belongs at the lowest level of the graph separator hierarchy, so to be able to traverse the overlay graph $H$ during each iteration of the SALT-p2p algorithm. 
Note it is sufficient to store only the $c^{1}(o)$, since cell IDs for higher levels may be calculated from that. 

\begin{wrapfigure}{l}{0.35\textwidth}
  \centering
    \includegraphics[width=0.25\columnwidth]{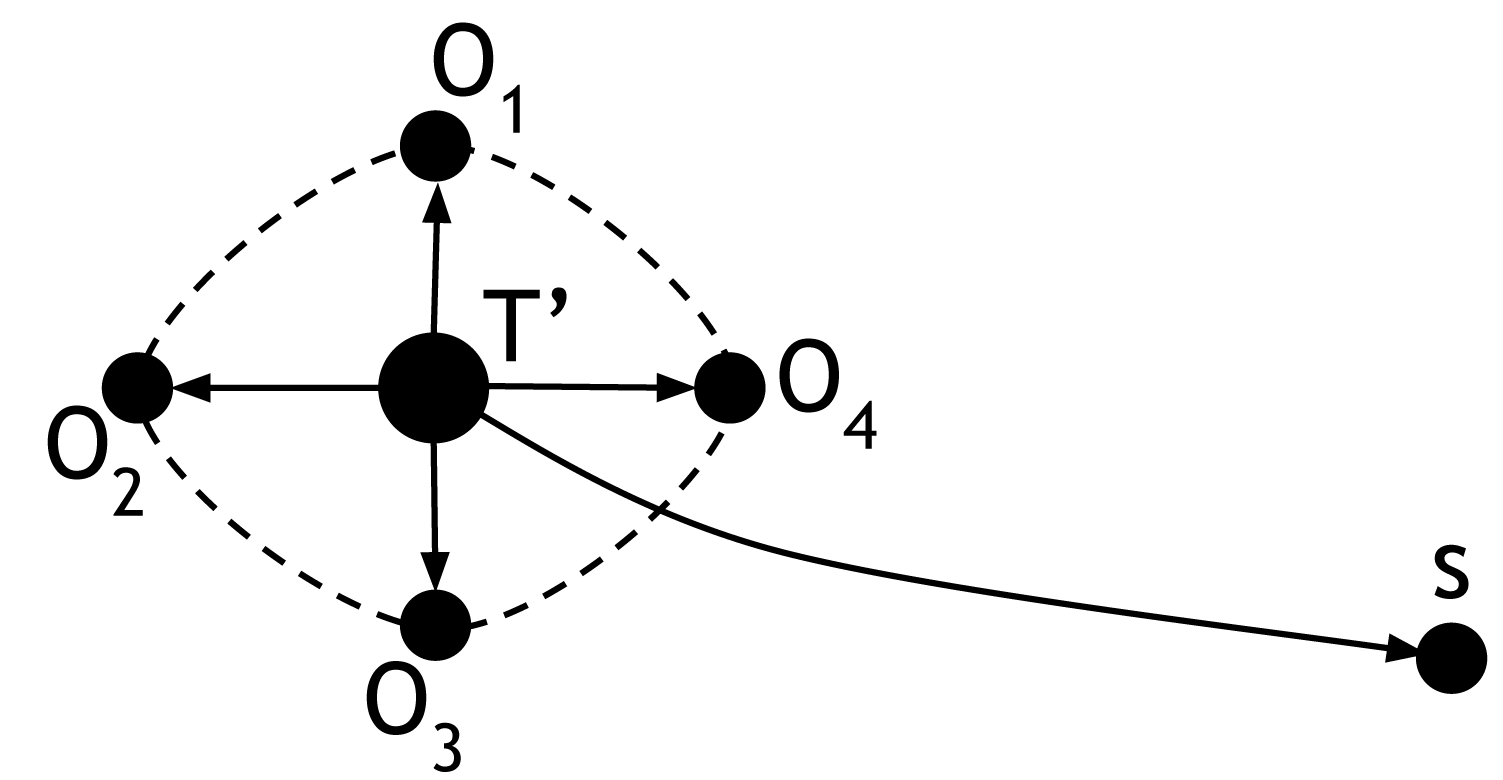}
  \caption{The $i$-th iteration of the SALT-kNN algorithm}
    \label{fig:salt-knn}
\end{wrapfigure}

Although the SALT-\kNN\ algorithm will be very fast for retrieving the first NN result object, it will become progressively slower when retrieving the additional $k-1$ NN vertices, since at each iteration, the SALT-p2p algorithm will start from scratch. 
To remedy this, at the beginning of the $i$-th iteration, we reload the corresponding priority queue with all vertices labeled during the $i{-}1$ iteration except those originating from the previous NN vertex found, since most of those labeled vertices were already assigned correct SP distances. 
For previously labeled vertices of which the SP distance can be improved, by using a min-heap priority queue (as all Dijkstra variants), the $i$-th iteration of the algorithm will further assign correct SP distances. 
This optimization significantly improves query times and still ensures correctness of the SALT-\kNN\ algorithm.

\subsection{Single-source shortest-path queries}
\label{subsec:sssp_queries_contrib}

Although our main contributions are
the SALT-p2p and the SALT-kNN algorithms for handling SPSP and $k$-NN queries, the SALT framework may still be used for other types of single-source shortest-path (SSSP) queries, including one-to-all, one-to-many and range queries by using the GRASP, reGRASP and isoGRASP algorithms presented in \cite{efentakis2014}. 
The major improvement in SALT is that by tweaking the number of cell levels ($L=6$) and the number of cells at the upper cell level ($|C^L|=16$), \emph{we may efficiently answer both forward and reverse SSSP queries} without increasing the preprocessing time significantly. As shown in previous works \cite{efentakis2013b}, this type of flexibility is extremely important for a wide range of geomarketing applications.

\subsection{SALT Tuning}
\label{subsec:salt_tuning}

Although SALT correctness for both p2p and \knn\ queries was rather straightforward to prove, we describe some of the necessary optimizations for an efficient SALT implementation. These optimizations include the most optimal schema for storing shortest-path distances of graph vertices \emph{from} and \emph{to} the landmarks, the use of SIMD instructions during the query phase of the SALT-kNN and SALT-p2p algorithms and how we assign IDs to vertices for fewer cache misses and acceleration of SALT's preprocessing and query phases. In detail:

\vspace{5 pt}
\textbf{Landmark Distance Records}. In \cite{efentakis2013d} landmark distances were stored in a 32-bit vector of size $2{\cdot}|S|{\cdot}|V|$. The distance of node with nodeID $i\in[0,|V|{-}1]$ from landmark number $j\in[0,|S|{-}1]$ was stored at position $2{\cdot}|S|{\cdot} i+2j$ and the distance of node $i$ to the landmark $j$ was stored in the next position ($2{\cdot}|S|{\cdot} i+2j+1$). Moreover, landmark distances \emph{from landmarks to nodes} were stored negated (as negatives), since this is how they are used for estimating lower-bounds. 
Although this storage schema facilitates fast calculation of lower-bounds, it is not optimal for calculating upper-bounds, as needed during the pruning phase of the SALT-kNN algorithm. 
To calculate the upper-bound of the distance between the query location $s$ and any object in $O$ (cf. Equation~\ref{eq:upper_bound}) we only need the distances from the object \emph{to} the landmarks and the distances \emph{from} the landmarks to the query location $s$. Thus, it is better to store landmark distances \emph{from} all landmarks on consecutive memory locations per vertex (and negated as before) and then the distances \emph{to} all landmarks. 
Hence, we use again a 32-bit vector of size $2{\cdot}|S|{\cdot}|V|$ for storing the landmark distances, but now the distance of vertex $i$ \emph{from} landmark $j$ is stored at position $2{\cdot}|S|{\cdot}i+j$ and the distance of node $i$ \emph{to} the landmark $j$ is stored in the position ($2{\cdot}|S|{\cdot} i+|S|+j$). 
With this optimization, to calculate the upper-bounds, we access $|S|$ consecutive memory locations per object instead of $2|S|$. Also, since landmark distances to vertices are stored negated, instead of addition, we use subtraction during the upper-bound calculation.  

\vspace{5 pt}
\textbf{Node Reordering}. Similar to previous works ~\cite{Baum2013, efentakis2014}, assigning smaller IDs to border vertices of higher levels and breaking ties within the same level by cell, has shown to improve performance for both SALT-p2p and SALT-\kNN\ algorithms. As a result, this is also the node-ordering of choice for our SALT framework.

\vspace{5 pt}
\textbf{SIMD Instructions}. Current x86-CPUs have special 128-bit SSE registers that hold four 32-bit integers and allow basic operations, such as addition, minimum, and maximum to be executed in parallel. 
Efentakis et al.~\cite{efentakis2013d} have utilized those 128-bit SSE registers to significantly accelerate the computation of the landmarks based lower-bounds.
We further expand this optimization by 
applying the above method to the efficient calculation of upper-bounds as well. Consequently, the pruning phase of the SALT-kNN algorithm requires significantly less than $1ms$. To the best of our knowledge, this work is the \emph{first to utilize SIMD instructions within the context of \knn\ queries}.

\subsection{Summary and Expectations}
\label{subsec:motivation}

Although our experimentation (see Sec.~\ref{sec:experiments}) will show that SALT is very efficient for all types of shortest-path queries, the main phase of SALT-\kNN\ could be performed with any valid unidirectional SP algorithm. 
The use of SALT-p2p has multiple advantages 
(i)~Its constituent algorithms, ALT and CRP, are the only algorithms with fast enough preprocessing times to be used for the case of \emph{dynamic road networks}. 
SALT-p2p ``inherits'' this important property necessary for providing the optimal algorithmic foundation for live traffic-based services.
(ii)~The pruning phase of SALT-\kNN\ is very crucial for a fast implementation. \emph{Only the landmarks preprocessing data could provide this type of functionality} that could potentially replace R-tree based approaches in other location-based services as well. 
(iii)~SALT-p2p is very \emph{robust with respect to the metric used}. In fact, its query performance is slightly better for travel distances, the metric for which most hierarchical SP approaches perform badly. This is an important property for \knn\ queries identifying Points-Of-Interest based on walking distance. 
(iv)~Our results show that unidirectional SALT-p2p actually provides better performance than bidirectional SALT-p2p. This is an advantage over existing hierarchical methods, since most can only be used in a bidirectional setting. 
(v)~Last but not least, the main phase of the SALT-\kNN\ algorithm initially expands vertices closer to the query location. As such, ``unattractive'' objects furthest from the query location (as estimated by the lower-bounds) that cannot be excluded during the pruning phase do not slow down SALT-\kNN\ queries. In fact, our experimentation will show that finding the first nearest neighbor is almost as fast as a plain SALT-p2p query. Hence, it is hard to provide a significantly better theoretical solution, at least using standard SP techniques, with fast enough preprocessing times suitable for dynamic road networks.


\section{Experiments}
\label{sec:experiments}

The experimentation that follows, assesses the performance of the SALT framework and 
the respective SALT-p2p and SALT-kNN algorithms. For completeness, we also report the performance of sequential and parallel GRASP~\cite{efentakis2014} algorithm within the SALT framework for single-source (one-to-all) queries. 

Experiments were performed on a workstation with a four-core i7-4771 processor clocked at 3.5GHz with 32 GB of RAM, running Ubuntu 14.04 64bit. Our code was written in C++ and compiled with GCC 4.8 (and OpenMP). Query times are executed on one core and augmented with SSE instructions. We used the 
European road network with 18M nodes~/~42M arcs and the full USA road network with 24M nodes~/~58M arcs \cite{dimacs2009} and
experimented with both travel times and travel distances. 

For partitioning the graph into nested-multilevel partitions, similarly to \cite{efentakis2014}, we used Buffoon~/~KaFFPa \cite{sanders2011e} in a top-down approach. We use a partitioning setup similar to the best recorded CRP results of \cite{delling2013w} with total number of overlay levels set to $L{=}6$ and $|C^1|{=}1048576$, $|C^2|{=}65536$, $|C^3|{=}8192$, $|C^4|{=}1024$, $|C^5|{=}128$ and $|C^6|{=}16$. We also used 24 landmarks, since adding more landmarks did not offer significant performance benefits for either SALT-p2p or SALT-kNN algorithms.

\subsection{Preprocessing}
\label{subsec:preprocessing_exp}  

\begin{wraptable}{l}{0.58\textwidth}
\centering
\caption{SALT, GRASP and G-tree preprocessing times}
{\footnotesize
\begin{tabular}{|c|c|c|c|c|}
\hline
&   \multicolumn{4}{c|} {\bfseries{Preprocessing time (s)} }\\ 
\hline
&   \multicolumn{2}{c|} {\bfseries{Travel Times (TT)}}&\multicolumn{2}{c|} {\bfseries{Travel Distances (TD)}} \\ 
\hline
& \bfseries{EUR} & \bfseries{USA}&\bfseries{EUR} & \bfseries{USA}\\ \hline
\bfseries{SALT (GS custom. phase)} &11.1 (5.5) &14.82 (7.4) &11.3 (5.7) &15.4 (7.7) \\ \hline
\bfseries{SALT (Landmarks phase)} &2.6 (1.3) &3.6 (1.8) &2.7 (1.4) &3.6 (1.8)\\ \hline
\bfseries{SALT (Total)} &13.7 (6.9) &18.4 (9.2) &14.0 (7.0) &18.9 (9.5) \\ \hline
\bfseries{GRASP (Orig)} &8 (8)&12 (12)&10 (10)&13 (13)\\ \hline
\bfseries{G-tree} &(198,479)&(5,736)&(25,918)&(5,001)\\ 
\hline
\end{tabular} 
 \label{tab:prep}
}
\end{wraptable}

In this section we will report the preprocessing times for SALT, in comparison to the original GRASP version (as reported on ~\cite{efentakis2014}) and G-tree~\cite{zhong2013} (G-tree source code was kindly provided by its authors).  Note, that contrary to the SALT framework that may simultaneously answer single-pair, single-source (one-to-all, one-to-many, range) and $k$-NN queries, GRASP only focuses on single-source queries and G-tree may only be used for undirected networks and $k$-NN queries. 
SALT and GRASP preprocessing times refer to parallel execution, using  all available cores and G-tree preprocessing times are sequential. For GRASP and SALT and its graph-separator sub-phase we only report preprocessing times for the customization stage, similar to \cite{delling2011} and \cite{efentakis2014}, since this is the preprocessing that must be repeated when edge-weights change, as in the case of live-traffic road networks. For a fair comparison, for G-tree we do not report the partitioning time required for the building of the G-tree index (which uses METIS \cite{metis})  and we only report the preprocessing time for calculating the SP distances between the vertices inside the respective index structure. Results are presented on Table~\ref{tab:prep}. Numbers inside parentheses represent preprocessing times for undirected versions of the road networks. 

Results clearly show that: (i)~G-tree preprocessing times are very disappointing, especially for Europe and travel times, when more than 24h are required for preprocessing, 
which is in huge contrast with SALT's preprocessing time \emph{which never exceeds 19s} for all networks and metrics. 
(ii)~In comparison to GRASP, SALT may calculate both forward and reverse graph SSSP queries. If GRASP was to be extended for reverse graph SSSP queries, its preprocessing time would double and hence \emph{it would be $16{-}43\%$ slower than SALT}. 
(iii)~SALT's preprocessing time \emph{is very robust to the metric used} and preprocessing time is similar for both metrics. (iv)~For undirected versions of the road networks (for comparing results to G-tree), SALT's preprocessing time drops in half, both for the GS customization and landmarks phase.  Note that although SALT's total preprocessing time is better than any other previous ALT based approach including \cite{efentakis2013d}, the GS customization phase could be potentially further accelerated by using the optimizations of \cite{delling2013e}, such as SIMD instructions or contraction.
But even without those potential optimizations, SALT still provides excellent preprocessing time, considering the fact that SALT may answer all variants of shortest-path queries on road networks. 

\subsection{Single-pair / single-source shortest-path queries}
\label{subsec:p2p_sp_queries_exp}  

In this section we will describe unidirectional and bidirectional SALT-p2p query performance for single-pair shortest-path (SPSP) queries, compared to its individual algorithmic components, namely ALT \cite{goldberg2005} as augmented in~\cite{efentakis2013d} and our customized CRP \cite{delling2011} with the arc-reduction of~\cite{efentakis2012}, within the SALT framework. To that purpose, we executed 10,000 point-to-point queries with the pair of vertices selected uniformly at random. Regarding single-source shortest-path (SSSP) queries, we report sequential and parallel performance of GRASP for one-to-all queries within the SALT framework and compare it with the  original version of GRASP (as reported on ~\cite{efentakis2014} on an almost identical setting to ours). 
For both GRASP versions, the number in parentheses represent sequential times. Results are presented in Table~\ref{tab:spsp}. 

Considering SPSP query performance, results show that: (i)~Unidirectional SALT-p2p is always faster than bidirectional SALT-p2p. This is in stark contrast with its individual components (ALT and CRP), in which bidirectional performance is significantly better. 
Thus, to the best of our knowledge, \emph{uniSALT-p2p is the faster unidirectional algorithm for road networks, with preprocessing times of  few seconds}. (ii) SALT-p2p is $100{-}266$ times faster than ALT and 
$3{-}4$ times faster than CRP. Note that our CRP's query performance is almost identical to the best CRP implementation of \cite{delling2013w}. Moreover, SALT-p2p path unpacking (i.e., providing full paths) would also be faster than CRP, since it uses bidirectional ALT instead of bidirectional Dijkstra used by CRP \cite{delling2013w}. (iii)~SALT-p2p is impressively robust to the metric used. In fact, \emph{uniSALT-p2p is slightly faster when we switch from travel times to travel distances}.

\begin{wraptable}{l}{0.55\textwidth}
\centering
\caption{SALT-p2p and GRASP query performance}
{\footnotesize
\begin{tabular}{|c|c|c|c|c|}
\hline
&   \multicolumn{4}{c|} {\bfseries{SPSP Query times (ms)} }\\ 
\hline
&   \multicolumn{2}{c|} {\bfseries{Travel Times (TT)}}&\multicolumn{2}{c|} {\bfseries{Travel Distances (TD)}} \\ 
\hline
& \bfseries{EUR} & \bfseries{USA}&\bfseries{EUR} & \bfseries{USA}\\ \hline
\bfseries{biALT} &103&60&133&89 \\ \hline
\bfseries{CRP (+AR)} &1.6&1.8&2&2 \\ \hline
\bfseries{uniSALT-p2p} &0.6&0.6&0.5&0.5 \\ \hline
\bfseries{biSALT-p2p} &0.9&0.9&0.9&0.9\\ \hline
&   \multicolumn{4}{c|} {\bfseries{SSSP Query times (ms)} }\\ \hline
\bfseries{GRASP (Orig)} &43 (150)&58 (207) &46 (156) &66 (218) \\ \hline
\bfseries{GRASP (SALT)} &50 (169) &65 (224)&53 (175) & 68 (228) \\ \hline
\end{tabular} 
 \label{tab:spsp}
}
\end{wraptable}

Regarding SSSP queries, the GRASP implementation within the SALT framework is $5{-}12\%$ slower for sequential and $3{-}16\%$ slower for parallel execution than the original GRASP implementation. Still, it is fast enough for most practical cases and the SALT framework may also execute forward and reverse SSSP queries, which is also a considerable advantage. Note, that the slightly less efficient implementation of GRASP within SALT is mainly attributed to the fact that now $|C_L|=16$ (in comparison to $|C_L|=128$ in the original paper). Still, setting $|C_L|=16$ is the optimal setting for SPSP and $k$-NN queries, which constitute the most typical queries encountered in any shortest-path framework on road networks. In this sense, we decided to use this setting that benefits the most frequent type of queries. 

\subsection{$k$-NN queries}
\label{subsec:$k$-NN_queries3}

In this section, we compare performance between SALT-kNN, Dijkstra and G-tree~\cite{zhong2013}, in the context of $k$-NN queries. For each experiment we generate $100$ sets of random objects of varying size $|O|$  and for each such set we generate $100$ random query locations, for a total of $10,000$ $k$-NN queries per $|O|$. The same $10,000$ queries per $|O|$ are executed for varying values of $k =\{1,2,4,8,16\}$ and we report average query times. 
Note, that G-tree also requires a \emph{target selection phase}, for each set of objects $|O|$ (which takes almost $1.9{-}2.4 s$). Thus, contrarily to both Dijkstra and SALT-kNN, G-tree cannot be used for moving objects. 
Results for $k=1$ and $k=4$ 
are presented in Fig.~\ref{fig:k_4}. 

\begin{figure}[!tb]
\centering
 \subfigure[Europe ($k{=}1)$]
{\includegraphics[width=0.305\columnwidth]{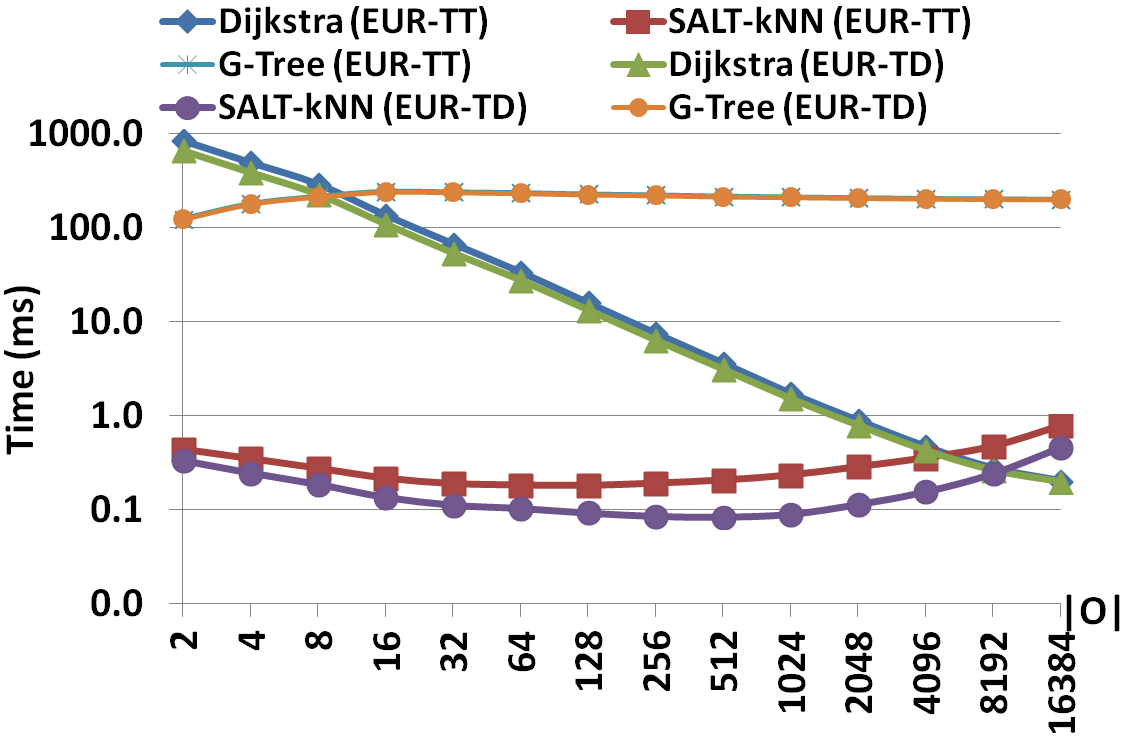}\label{fig:salt_eur_k1}}\qquad \qquad
 \subfigure[USA ($k{=}1)$]
{\includegraphics[width=0.305\columnwidth]{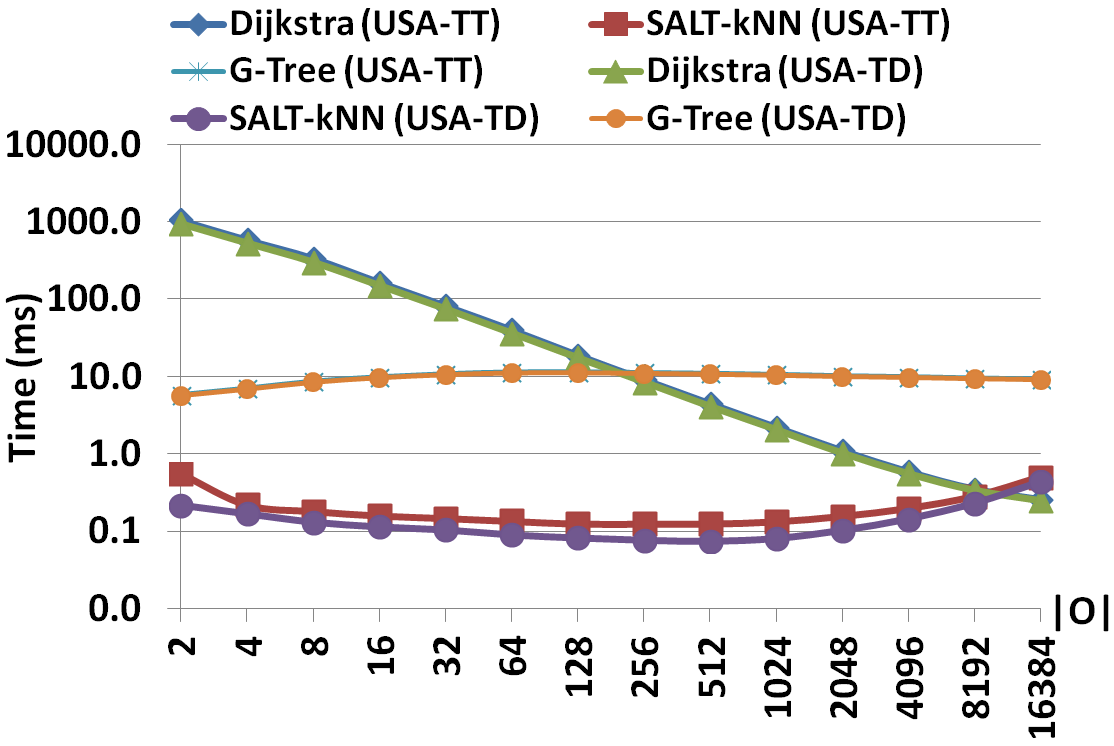}\label{fig:salt_usa_k1}}
 \subfigure[Europe ($k{=}4$)]
{\includegraphics[width=0.305\columnwidth]{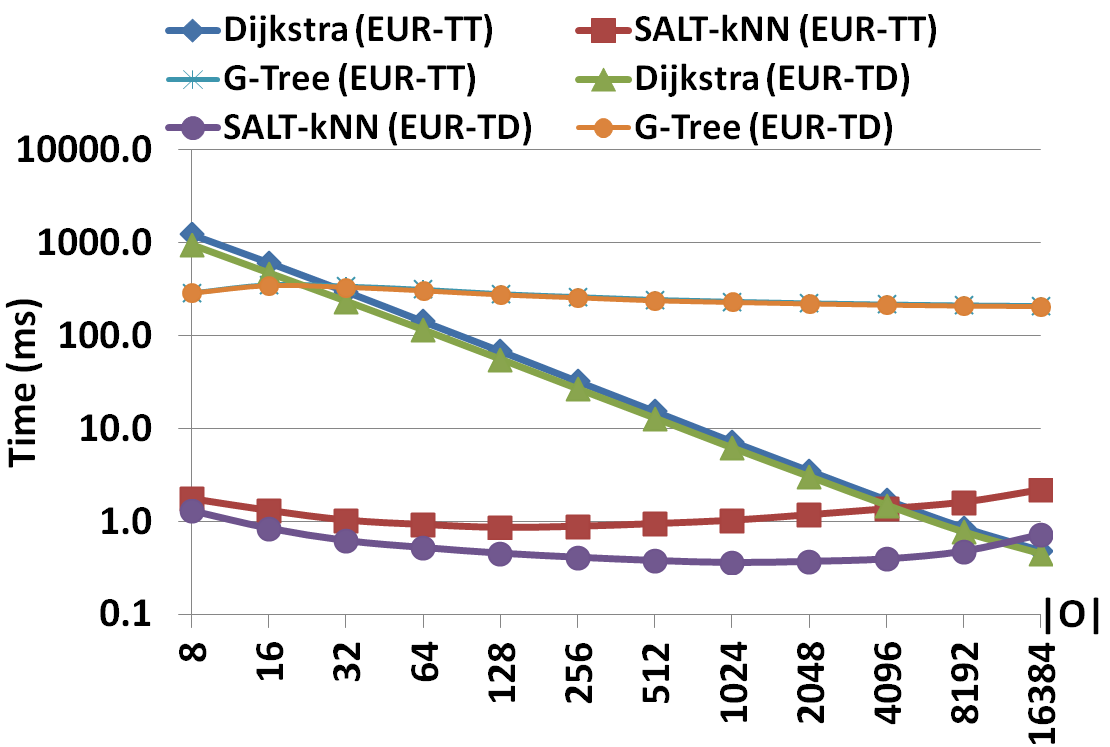}\label{fig:salt_eur_k4}}\qquad \qquad
 \subfigure[USA ($k{=}4$)]
{\includegraphics[width=0.305\columnwidth]{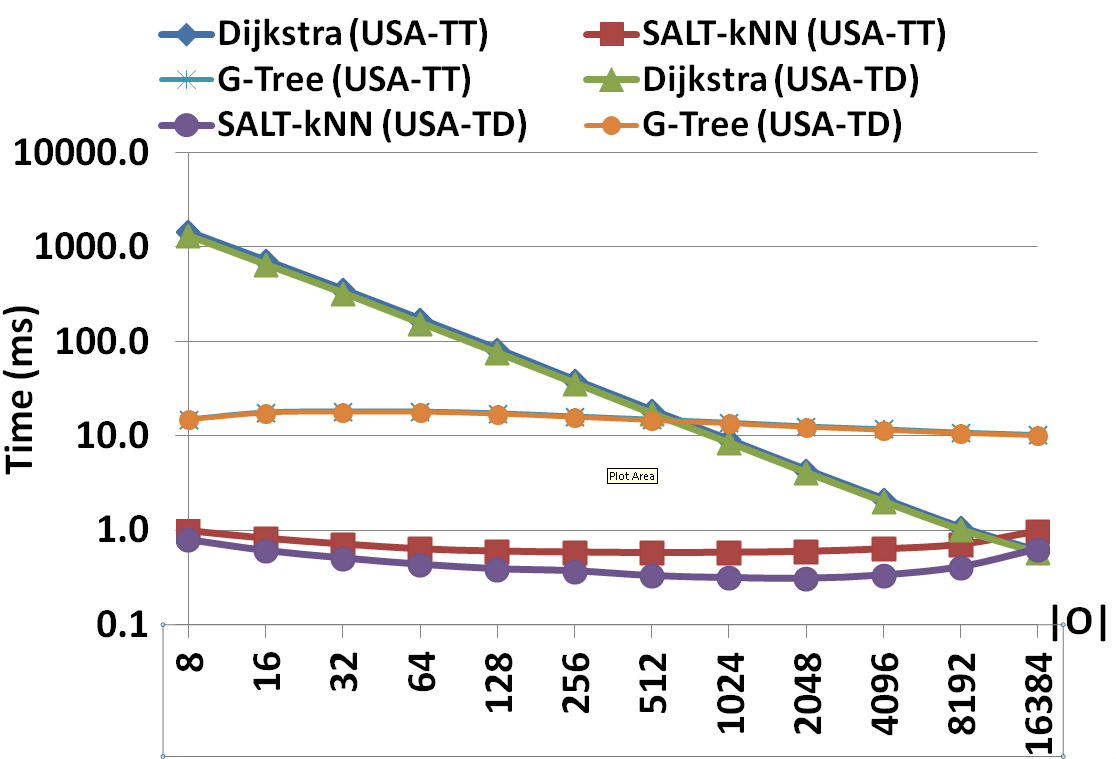}\label{fig:salt_usa_k4}}
 \caption{SALT-kNN, Dijkstra and G-tree comparison for $k{=}1$ and  $k{=}4$ and varying values of $|O|$.}
 \label{fig:k_4}
\end{figure}  

Results clearly show that SALT-kNN provides stable performance and query times significantly below $1ms$ for $k{=}1$. Contrarily, G-tree is almost \emph{two - three orders of magnitude slower} and therefore cannot compete with either SALT-kNN or Dijkstra. Dijkstra starts very slow for small values of $|O|$ but manages to surpass SALT-kNN performance for $|O|>8192$.   
This was inevitable to happen for any $k$-NN method, since \emph{if the number of objects is significantly large and their distribution is random, an efficient Dijkstra implementation would only have to scan a few hundred nodes}. 
Still, since for static points of interest we are usually interested in a specific type of objects (e.g., gas stations) and in the case of moving objects we rarely have such large vehicle fleets (i.e., taxis, trucks) to monitor and 
we usually aim for $k$-NN queries among the \emph{available} vehicles (a much smaller subset of total vehicles), then the SALT-kNN algorithm is surely to perform better for most practical applications. 

After establishing the superiority of SALT-kNN over G-tree, in our second set of experiments, we evaluate the impact of objects distribution to SALT-kNN and Dijkstra's performance. To that purpose, we adapt the methodology of~\cite{delling2011g}. We pick a vertex at random and run Dijkstra's algorithm from it until reaching a predetermined number of vertices $|B|$. If $B$ is the set of vertices visited during this search, we pick our objects $O$ as a random subset of $B$. We keep the number of objects $|O|$ steady at $2^{14}$ and we experiment with different values of $|B|$ ranging from $2^{14} \dots 2^{24}$, to simulate cases of either: (i) points of interest mainly located near the city-center or (ii) vehicle fleets which may service an entire continent but operate mainly on a particular country. Results for $k=1$ and $k=4$ are presented in Fig.
~\ref{fig:k_1_14_24}.

\begin{figure}[!tb]
\centering
 \subfigure[Europe ($k{=}1$)]
{\includegraphics[width=0.305\columnwidth]{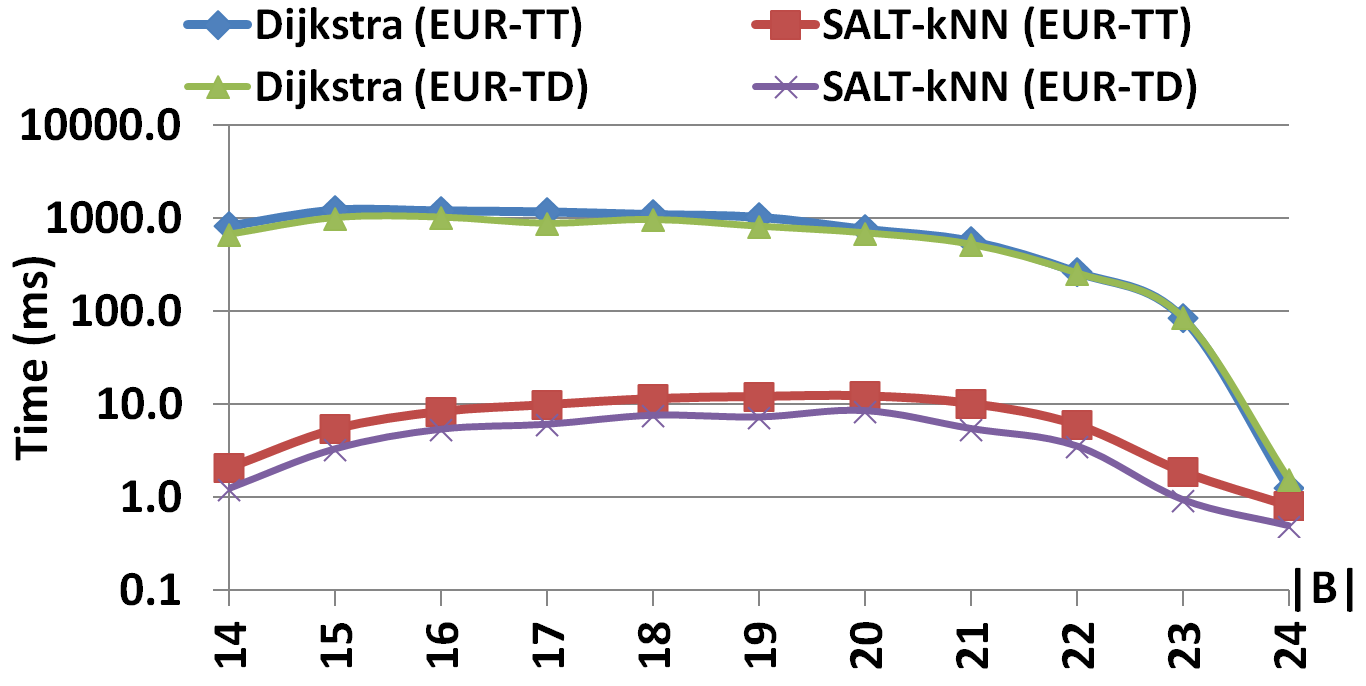}\label{fig:salt_eur_k1_14_24}}\qquad \qquad
 \subfigure[USA ($k{=}1$)]
{\includegraphics[width=0.305\columnwidth]{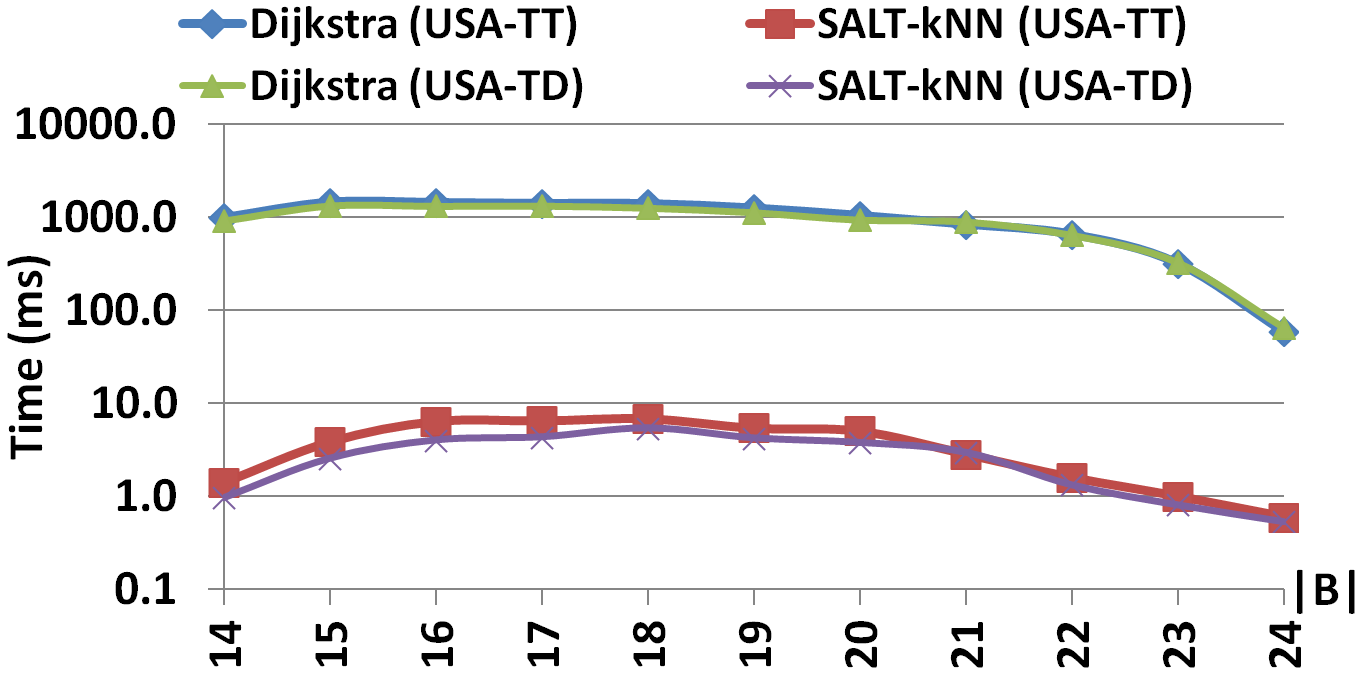}\label{fig:salt_usa_k1_14_24}}
 \subfigure[Europe ($k{=}4$)]
{\includegraphics[width=0.305\columnwidth]{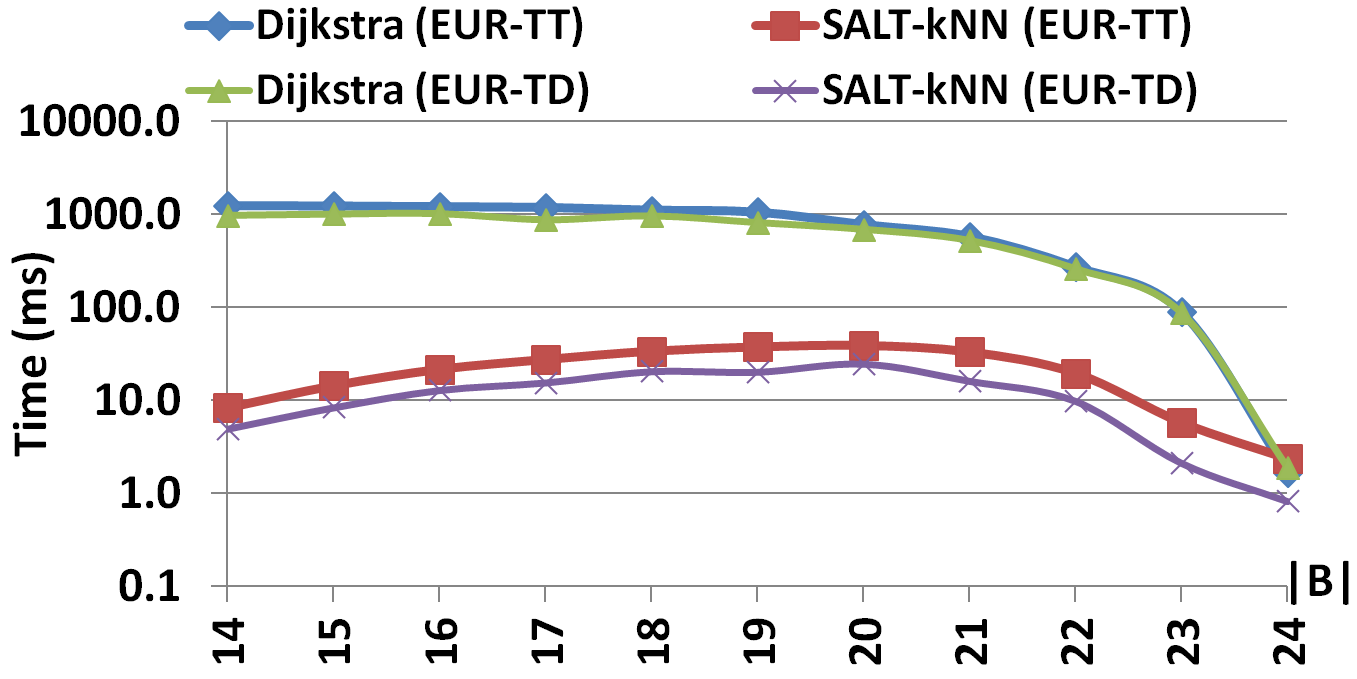}\label{fig:salt_eur_k4_14_24}}\qquad \qquad
 \subfigure[USA ($k{=}4$)]
{\includegraphics[width=0.305\columnwidth]{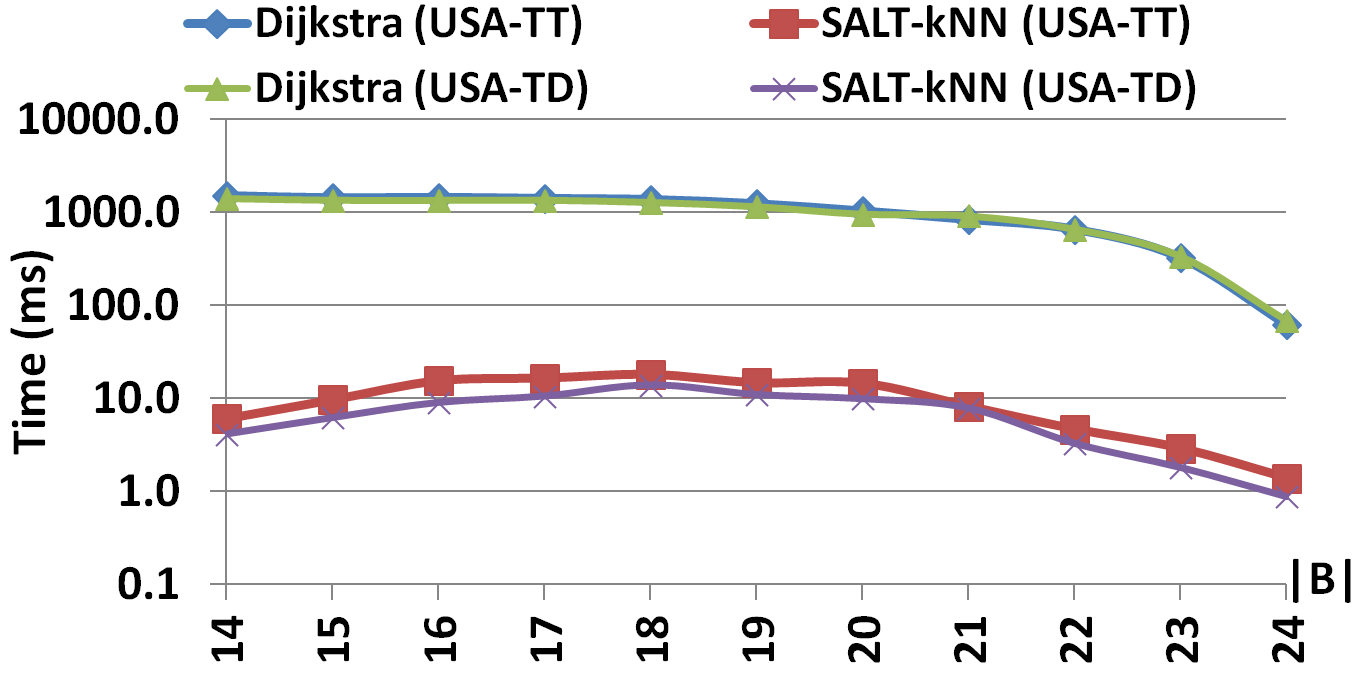}\label{fig:salt_usa_k4_14_24.png}}
 \caption{SALT-kNN and Dijkstra comparison for $|O|=2^{14}$, $k=1$ and $k=4$ and varying values of $|B|$.}\label{fig:k_1_14_24}
\end{figure}  

Result show, that once again SALT-kNN provides excellent performance regardless of the object distribution, contrarily to Dijkstra which is 1-2 orders of magnitude slower when objects are not uniformly located in the road network (which is the typical case, either for static or moving objects). Thus, SALT-kNN is the only algorithm that guarantees excellent and stable performance, regardless of: (i)~the number of objects and (ii)~the objects distribution. Moreover, it does not need a target selection phase, such as G-tree or CRP and therefore, it may be used for either static or moving objects. Note, than even without building an index, CRP would still require $10ms$ for the target selection phase for 16384 objects for the Europe road network (as recorded in \cite{delling2014pp} on a better workstation than ours) and therefore, CRP would \emph{be at least 10 times slower than SALT-kNN for moving objects}. 


\section{Summary and Conclusions}
\label{sec:conclusions}

This work presented SALT, a novel framework for answering shortest-path queries on road networks, including point-to-point, single-source (one-to-all, one-to-many, range) and \knn\ queries. By combining ideas from the ALT, CRP and GRASP algorithms, the SALT framework efficiently answers point-to-point queries $3{-}4$ times faster than previous algorithms of similar preprocessing times and answers \knn\ queries orders of magnitude faster than previous index-based approaches. Moreover, the proposed SALT-kNN algorithm was shown to be especially robust,  regardless of the metric used, the number of objects or the distribution of objects in the road network. Hence, it presents itself as an excellent solution for most practical use-cases. 

Despite its excellent query performance, the most important advantage of the SALT framework is its flexibility and versatility with respect to the different variants of the shortest-path queries it services. 
The exact same data structures efficiently tackle a wide range of different shortest-path problems, with preprocessing time of only a few seconds, making SALT suitable for dynamic (live-traffic) road networks as well. To the best of our knowledge, there is no other framework that matches the benefits and versatility of SALT. We truly consider it the algorithmic version of a swiss army knife for shortest-path queries and the best overall solution for real-world applications. 

\section*{Acknowledgments}
{\small The research leading to these results has received funding from the EU FP7 Project  
 ``GEOSTREAM'' (\url{http://geocontentstream.eu}, GA: FP7-SME-2012-315631). The authors would also like to thank Ruicheng Zhong for providing us access to the G-tree source code.}

\vspace{-3pt}
\bibliographystyle{abbrv}
\bibliography{gis2014}

\end{document}